\documentclass[10pt, final, twocolumn, twoside,romanappendices]{IEEEtran}
\usepackage{subfig}

\usepackage{winsnotation}
\usepackage[bookmarks,colorlinks]{hyperref} % create hyperlinks
\usepackage{color}
\hypersetup{colorlinks,citecolor= red,filecolor= blue,linkcolor= blue,urlcolor=blue}
\usepackage{graphicx}
\usepackage{tabularx,booktabs}
\newcolumntype{C}{>{\centering\arraybackslash}X} % centered version of "X" type
\usepackage{algorithm}
\usepackage[noend]{algpseudocode}
\usepackage{dsfont}
\usepackage{amsmath,amssymb}
\usepackage{mathtools}
\usepackage{graphicx}
\usepackage[font=small]{caption}
\usepackage[noadjust]{cite}
\usepackage{amsmath,amsfonts,amssymb,amsthm}
\usepackage{mathtools}
\usepackage{commath}
\usepackage{acronym}  % make an acronym
\acrodef{ofdm}[OFDM]{Orthogonal Frequency-Division Multiplexing}
\acrodef{fft}[FFT]{ fast Fourier transform}
\acrodef{iot}[IoT]{Internet of Things}
\acrodef{mimo}[MIMO]{multiple-input multiple-output}
\acrodef{siso}[SISO]{single-input single-output}
\acrodef{clt}[CLT]{ Central Limit Theorem}
\acrodef{ibi}[IBI]{inter block interference}

\acrodef{ofdma}[OFDMA]{Orthogonal Frequency-Division Multiple Access}

\acrodef{isi}[ISI]{inter symbol interference}
\acrodef{cp}[CP]{cyclic prefix}
\acrodef{zp}[ZP]{zero padding}
\acrodef{fir}[FIR]{finite impulse response}
\acrodef{v2x}[V2X]{Vehicle-to-everything}
\acrodef{nda}[NDA]{non-data-aided}
\acrodef{da}[DA]{data-aided}
\acrodef{ml}[ML]{maximum likelihood}
\acrodef{to}[TO]{timing offset}
\acrodef{wed}[WED]{Weighted Energy Detector}
\acrodef{ed}[ED]{Energy Detector}
\acrodef{rms}[RMS]{root mean square}
\acrodef{tm}[TE]{transition estimator}
\acrodef{ge}[GE]{Gamma estimator}
\acrodef{pdf}[PDF]{probability density function}
\acrodef{cfo}[CFO]{carrier frequency offset}

\acrodef{iid}[i.i.d]{independent and identically distributed}

\acrodef{bem}[BEM]{basis expansion model}
\acrodef{ls}[LS]{least squares}
\acrodef{mmse}[MMSE]{minimum mean square error}
\acrodef{pa}[PA]{pilot-aided}
\acrodef{dd}[DD]{decision-directed}
\acrodef{cc}[CE]{channel estimation}
\acrodef{dnn}[DNN]{deep neural network}
\acrodef{mse}[MSE]{mean-squared error}
\acrodef{dl}[DL]{deep learning}
\acrodef{ci}[CI]{channel state information}
\acrodef{mmse}[MMSE]{minimum mean square error}
\acrodef{awgn}[AWGN]{additive white Gaussian noise}
\acrodef{map}[MAP]{maximum a posteriori probability}
\acrodef{ber}[BER]{bit error rate}
\acrodef{kf}[KF]{Kalman filter}
\acrodef{snr}[SNR]{signal-to-noise ratio}
\acrodef{iot}[IoT]{Internet of Things}

\acrodef{chf}[CHF]{Characteristic function}

\newcommand{\bd}{\begin{description}}
\newcommand{\ed}{\end{description}}
\newcommand{\be}{\begin{enumerate}}
\newcommand{\ee}{\end{enumerate}}
\newcommand{\bi}{\begin{itemize}}
\newcommand{\ei}{\end{itemize}}
\newcommand{\bl}{\begin{list}}
\newcommand{\el}{\end{list}}
\newcommand{\bt}{\begin{tabbing}}
\newcommand{\et}{\end{tabbing}}

\usepackage{amssymb}
\usepackage[yyyymmdd,hhmmss]{datetime} % gives time of the compiling
\newdateformat{monthyeardate}{\monthname[\THEMONTH] \THEDAY, \THEYEAR} % specifies the format of the date
\usepackage [autostyle, english = american]{csquotes}
\MakeOuterQuote{"}
\newtheorem{theorem}{Theorem}
\newtheorem{lemma}{Lemma}[theorem]
\newtheorem{corollary}{Corollary}[theorem]

\usepackage[usestackEOL]{stackengine}
\stackMath
\usepackage{stfloats}
\usepackage{mathtools}
\usepackage{arydshln}
\usepackage{booktabs}% More proffesional look of tables.
\usepackage{siunitx}% An awesome package for typesetting and manipulation numbers and units.

\begin{document}

\title{An Approximate Maximum Likelihood Time Synchronization Algorithm for Zero-padded OFDM in   Channels with Impulsive Gaussian Noise}
\author{
	\vspace{0.2cm}
   Koosha~Pourtahmasi~Roshandeh,~\IEEEmembership{Student~Member,~IEEE},
   Mostafa~Mohammadkarimi,~\IEEEmembership{Member,~IEEE}, and
    Masoud~Ardakani,~\IEEEmembership{Senior~Member,~IEEE}
}
\maketitle
%\markboth{PLEASE DO NOT DISTRIBUTE WITHOUT THE WRITTEN CONSENT OF THE AUTHORS \version}
%		{Pourtahmasi Roshandeh, Mohammadkarimi, and Ardakani: \paperTitleMarkboth}

\begin{abstract}
In wireless communication systems, \ac{ofdm} includes variants using either a  \ac{cp} or a \ac{zp} as the guard interval to avoid inter-symbol interference. 
\ac{ofdm} is ideally suited to deal with frequency-selective channels and \ac{awgn}; however, its performance may be dramatically degraded in the presence of impulse noise.
While the \ac{zp} variants of \ac{ofdm} exhibit lower \ac{ber} and higher energy efficiency compared to their \ac{cp} counterparts, they demand strict time synchronization, which is challenging in the absence of pilot and \ac{cp}. Moreover, on the contrary to \ac{awgn}, impulse noise severely corrupts data.
In this paper, a new low-complexity \ac{to} estimator for \ac{zp}-\ac{ofdm} for practical impulsive-noise environments is proposed, where 
relies on the second-other statistics of the multipath fading channel and noise. Performance comparison with existing  \ac{to} estimators demonstrates either a superior performance in terms of lock-in probability  or a significantly lower complexity over a wide range of \ac{snr} for various practical scenarios. 
\end{abstract}

%Specifically, this estimator assigns an approximate Gaussian distribution to the received samples for \ac{to} estimation

\section{Introduction}

\IEEEPARstart{O}{fdm} technique is widely employed in wireless communications, mainly due to its ability in converting a frequency-selective fading channel into a group of flat-fading sub-channels \cite{lu2000space}. 
Compared to conventional single-carrier systems, \ac{ofdm} 
offers  increased robustness  against  multipath fading distortions  since    channel equalization  can be easily  performed  in  the  frequency domain  through  a  bank  of  one-tap  multiplier \cite{dai2010positioning}. Moreover, \ac{ofdm} can be efficiently implemented using  \ac{fft} \cite{murphy2002low}, which makes it more appealing compared to other multi-carrier modulation techniques such as Filter Bank Multi Carrier and Generalised Frequency Division Multiplexing. 

 Owing to superior advantages of \ac{ofdm}, it is exploited in many IEEE standards, such as, IEEE 802.15.3a, IEEE  802.16d/e,  and IEEE 802.15.4g  \cite{green, jimenez2004design, ofdm2ieee802.15}, which are used for different applications. For instance, \ac{ofdm} combined with massive \ac{mimo} technique achieves a high data rate, making it suitable for multimedia broadcasting \cite{kim2008apparatus}. Moreover, many \ac{iot} applications such as smart buildings and \ac{v2x} leverage \ac{ofdm} as their main communication scheme \cite{ofdmieee802.15, ofdm2ieee802.15 }.
 
 \ac{ofdm}, however, undergoes a sever \ac{isi} caused by the high selectivity of the fading channel \cite{ wang2005robust}. In order to mitigate this issue, usually a guard interval with a fixed length is inserted between every two consecutive \ac{ofdm} symbols. When the guard interval is the partial repetition of the transmitting data samples, this scheme is called \ac{cp}-\ac{ofdm} \cite{channelestimationCP }. The primary benefit of \ac{cp}-\ac{ofdm} is the ease of \ac{to} estimation or equivalently estimating the starting point of \ac{fft}. This is referred to as time synchronization \cite{tufvesson1999time}, and is easily carried out by using \ac{cp} and its correlation with the data sequence. Despite the ease of time synchronization in \ac{cp}-\ac{ofdm}, it possesses some major disadvantages such as excessive power transmission which is due to the transmission of \ac{cp}. \ac{zp}-\ac{ofdm} \cite{muquet2000reduced}, where the guard interval is filled with zeros, overcomes this issue. However, the time synchronization, or equivalently  \ac{to} estimation, in \ac{zp}-\ac{ofdm} becomes a very difficult and complicated task. 
 
 There are two approaches in order to estimate  \ac{to}  in \ac{zp}-\ac{ofdm}. In the first approach which is called \ac{da} time synchronization, a series of training sequences (pilots) are used to estimate \ac{to}. The second approach, referred to as \ac{nda} time synchronization, however, relies on the statistical properties of the transmitted data sequence. 
 
 \subsection{Related work}
 
 The \ac{da} time synchronization for \ac{zp}-\ac{ofdm} has been studied in the literature \cite{nasir2016timing}. On the other hand, \ac{nda} time synchronization lacks a reliable mathematical analysis. An \ac{nda} time synchronization algorithm for \ac{zp}-\ac{ofdm} has been  proposed in \cite{bolcskei2001blind, LeNir2010} which are mainly  heuristic algorithms. More specifically, these algorithms trace the energy ratio of partially cropped data sequences; which  do not always  show a reliable performance in terms of lock-in probability for highly selective channels. Also, a mathematical approach towards \ac{nda} \ac{to} estimation for \ac{zp}-\ac{ofdm} systems has been proposed in \cite{koosha2020}. The authors in \cite{koosha2020} proposed a \ac{ml} \ac{to} estimator for a \ac{zp}-\ac{ofdm} under a frequency selective channel. However,  the algorithm in  \cite{koosha2020} is highly complex which hinders its implementation for \ac{mimo} systems. Moreover, the algorithm proposed in \cite{koosha2020} is deigned for \ac{awgn} channel which implies that its performance degrades when the channel experiences an impulsive noise.

\subsection{Motivation}

\ac{zp}-\ac{ofdm} requires less transmission power compared to \ac{cp}-\ac{ofdm}, due to lack of \ac{cp}, which makes it a suitable candidate for future \ac{iot} devices. However, time synchronization becomes challenging in \ac{zp}-\ac{ofdm} where time synchronization algorithms fail to achieve high lock-in probability, or are highly complex for practical implementations. Moreover, the proposed \ac{to} estimators for \ac{zp}-\ac{ofdm} systems so far are developed for Gaussian noise models. However, many real-world channels, e.g. underwater, urban and indoor channels, are known to experience an impulsive noise, rather than a simple Gaussian noise \cite{blackard1993measurements, middleton1973man, middleton1993elements,middleton1987channel}. This noise is originally coming from the great amount of noise in the nature, and the electronic equipment. It is well-known that designing communication systems under simple Gaussian noise model can significantly affect the performance  of such systems when they experience an impulsive noise in reality \cite{wang1999robust}. Hence, an accurate yet low-complex time synchronization algorithm for \ac{zp}-\ac{ofdm}   are still needed to be developed. 

In this paper, we propose a low complexity mathematical approach towards \ac{nda} time synchronization for \ac{zp}-\ac{ofdm} systems in an impulsive noise channel.
Simulation results demonstrate that the proposed estimator has a negligible performance gap in terms of lock-in probability  with the estimator in \cite{koosha2020}  while possessing a significantly lower complexity.

\subsection{Contributions}

In this paper, we 

\begin{itemize}

    \item propose a low-complexity approximate \ac{ml} \ac{to} estimator for \ac{mimo} \ac{zp}-\ac{ofdm} systems in highly selective channels with impulsive noise. This algorithm (i) achieves high lock-in probability, and (ii) has significantly lower complexity compared to \cite{koosha2020}.

    \item Higher order statistics of the proposed approximate \ac{pdf} of the received samples are investigated.
    
    \item A lower complexity algorithm, i.e. \ac{ed}, for \ac{mimo} \ac{zp}-\ac{ofdm} systems experiencing Gaussian noise is proposed. This algorithm achieves high lock-in probability for high \ac{snr}s while possessing a low complexity.

    \item Complexity  of the proposed algorithms and the one in \cite{koosha2020} has been studied.

\end{itemize}{}

%More specifically, we propose an approximate \ac{ml} \ac{to} estimator under Class A impulsive noise model in a time and frequency selective channel. Moreover, the proposed estimator is extended to \ac{mimo}-\ac{ofdm} systems which can be implemented in a fully vectorized fashion.  

The paper is organized as follows. System model is discussed in Section \ref{sec: sys mod}. The main ideas and the proposed \ac{ml} estimator for  systems are presented in Section \ref{sec: siso}. The complexity of the algorithm is compared to that of proposed in \cite{koosha2020} is studied in Section \ref{sec: complexity}. Simulation results and conclusions are given in Sections \ref{sec: simul}  and \ref{sec: conclu}, respectively.

\textit{Notations}: Column vectors are denoted by bold lower case letters. Random variables are indicated by uppercase letters. Matrices are denoted by bold uppercase letters. 
Conjugate, absolute value, transpose, and the expected value are indicated by $(\cdot)^*$, $|\cdot|$, $(\cdot)^{\rm{T}}$, and $\mathbb{E}\{\cdot\}$, respectively. Brackets, e.g. ${\bf a}[k]$,  are used for discrete indexing of a vector  ${\bf a}$.

\section{System model} \label{sec: sys mod}

We consider a \ac{mimo}-\ac{ofdm} wireless system with $m_t$ and $m_r$ transmit and receive antennas, respectively. This system uses \ac{zp}-\ac{ofdm} technique to  communicate over a frequency selective Rayleigh fading channel. We assume a perfect synchronization at the transmit antennas. Let 
$\{x^{(n)}_k\}_{k=0}^{n_{\rm{x}}-1}$, $\mathbb{E}\{|x_k|^2\}= \sigma^2_{\rm{x}}$ be the $n_{\rm{x}}$ complex data samples from the $n$-th \ac{ofdm} block to be transmitted  from the $i$-th transmit antenna.  Hence, their corresponding \ac{ofdm} signal can be expressed as 
\begin{align}\label{eq: ofdm symbol}
x^{(n)}_i(t)=\sum_{k=0}^{n_{{\rm{x}}-1}} x^{(n)}_k e^{\frac{j2\pi k t}{T_{\rm{x}}}}\,\,\,\,\,\,\,\,\,\,\ 0 \le t \le T_{\rm{x}},
\end{align}
\noindent where  $T_{\rm{x}}$ denotes the duration of the data signal. To deal with the delay spread of the wireless channel, 
a zero-padding guard interval of length $T_{{\rm{z}}}$ is added to \eqref{eq: ofdm symbol}, in order to form the transmitted \ac{ofdm} signal. Hence, the $n$-th transmitted \ac{ofdm} signal from the $i$-th transmit antenna is given as 
\begin{align} \label{eq: s continues}
s^{(n)}_i(t)=
\begin{cases}
x^{(n)}_i(t)  \,\,\,\,\,\,\,\,\,\,\,\ 0 \le t \le T_{\rm{x}}  \\
0 \,\,\,\,\,\,\,\,\,\,\,\,\,\,\,\,\,\,\,\,\,\,\,\,\,\,\,   T_{\rm{x}} < t \le T_{\rm{s}},
\end{cases}
\end{align}
\noindent where $T_{\rm{s}}$ denotes the signal duration, and  $T_{\rm{s}}= T_{\rm{x}}+T_{\rm{z}}$.

Since the receiver demodulates the received signal and performs sampling,  it is more convenient from both practical and mathematical point of view to develop algorithms and perform the analysis in baseband and discrete format. To this end, we continue our analysis in discrete baseband format from now on. 

Let us  denote the sampling rate at the receiver by $f_{\rm{s}}=1/T_{\rm{sa}}$.  We assume the transmitted signal, i.e. $s^{(n)}_i(t)$, passes through a frequency selective channel with $n_{\rm{h}}$-taps. Let $\{h_{ji}[k]\}_{k=0}^{n_{\rm{h}}-1}$ denote the $k$-th channel tap between the transmit antenna $i$ and the received antenna $j$. The channel taps are assumed to be statistically independent complex Gaussian random variables with zero-mean, i.e. Rayleigh fading. The delay profile of the taps are given as 
\begin{align}\label{7u8i0000}
\mathbb{E}{\{}h[k]h^*[k-m]{\}}=\sigma_{{{\rm{h}}_k}}^2\delta[m],
\end{align}
$k=0,1,\dots, n_{\rm{h}}-1$. We assume the channel delay profile is known to the receiver.

%We assume the transmitted signal, i.e. $s^{(n)}_i(t)$, passes through a frequency selective channel with $n_{\rm{h}}$-taps, where each tap is a zero-mean Gaussian random variable. Let $\{h_{ji}[k]\}_{k=0}^{n_{\rm{h}}-1}$, $\mathbb{E}\{|h_{ji}[k]|^2\}= \sigma^2_{{\rm{h}}_k}$ denote the $k$-th channel tap between the transmit antenna $i$ and the received antenna $j$. We also assume the channel taps follow an %exponential power delay profile; %hence, 
%\begin{align}
%\mathbb{E}\{|h_{ji}[k]|^2\}= %\sigma^2_{{\rm{h}}_k}=\alpha %\exp\bigg{(}\frac{-k}{n_{\rm{h}}}%\bigg{)},
%\,\,\,\,\,\,\ k=0,1,\cdots, %n_{\rm{h}}-1,
%\end{align}
%\noindent where $\alpha$ is a %normalizing factor that guarantees the average energy of the channel taps sum to one. We assume the channel delay profile is known to the receiver.

%Since the receiver demodulates the received signal and performs sampling,  it is more convenient from a practical point of view to develop algorithms and perform the analysis in baseband and discrete format. To this end, we denote the sampling period of the receiver by $T_{\rm{sa}}$, and continue our analysis in baseband and discrete format. Hence, 

In the absence of synchronization error and \ac{isi},  the discrete received baseband vector is expressed as 
\begin{align}\label{Sys Model: matrix form conv 2}
{\bf Y}^{(n)}=
\begin{cases}
{\bf H} {\bf S}^{(n)} + {\bf W}^{(n)},  \,\,\,\,\,\,\,\,\,\,\,\ n \ge 0  \\
{\bf W}^{(n)}, \,\,\,\,\,\,\,\,\,\,\,\,\,\,\,\,\,\,\,\,\,\,\,\,\,\,\,\,\,\,\,\,\,\,\,\,\,\,  n<0,
\end{cases}
\end{align}
where 
${\bf H}$ denotes the discrete channel filter, and is defined as 
\begin{equation} \label{matrix H}
 {\bf H}=\begin{pmatrix}\vspace{0.2 cm} 
{\bf {\rm{H}}}_{11} & {\bf {\rm{H}}}_{12} & \cdots & {\bf {\rm{H}}}_{1{\rm{m_t}}}  \\ \vspace{0.2 cm}
{\bf {\rm{H}}}_{21} & {\bf {\rm{H}}}_{22} & \cdots & {\bf {\rm{H}}}_{2{\rm{m_t}}}  \\ \vspace{0.2 cm}
\vdots & \cdots & \ddots & \cdots  \\ \vspace{0.2 cm}
{\bf {\rm{H}}}_{{\rm{m_r}} 1} & {\bf {\rm{H}}}_{{\rm{m_r}} 2} & \cdots & {\bf {\rm{H}}}_{ {\rm{m_r}} {\rm{m_t}}},  \\
\end{pmatrix}
\end{equation}
\noindent where ${\bf {\rm{H}}}_{ji}$ is the lower triangular Toeplitz channel matrix between transmit antenna $i$ and the received antenna $j$, with first column  $[h_{ji}[0] \ h_{ji}[1] \ \cdots \  h_{ji}[n_{\rm{h}}-1] ~ \  0 \  \cdots \ 0]^\text{T}$. We set this matrix to be $n_{\rm{s}} \times n_{\rm{s}}$ where 
\begin{equation}
\begin{split}
 n_{\rm{s}} = n_{\rm{x}}+n_{\rm{z}},~~~
 n_{\rm{s}}\triangleq T_{\rm{s}}/T_{\rm{sa}}\\
 n_{\rm{x}}\triangleq T_{\rm{x}}/T_{\rm{sa}},~~~
 n_{\rm{z}}\triangleq T_{\rm{z}}/T_{\rm{sa}}
\end{split}{}
\end{equation}
where $n_{\rm{s}},n_{\rm{x}},$ and $n_{\rm{z}}$ denote the number of \ac{ofdm} signal samples, number of data samples, and the number of zero samples, respectively. Moreover,  ${\bf S}^{(n)}$, ${\bf Y}^{(n)}$, and ${\bf W}^{(n)}$ are defined as 
\begin{equation} \label{sys mod:  s y w}
{\bf S}^{(n)} =\begin{pmatrix}\vspace{0.2 cm}
{\bf s}^{(n)}_1   \\ \vspace{0.2 cm}
{\bf s}^{(n)}_2   \\ \vspace{0.2 cm}
\vdots   \\ \vspace{0.2 cm}
{\bf s}^{(n)}_{{\rm{m_t}}}   \\
\end{pmatrix},~
{\bf Y}^{(n)} \triangleq \begin{pmatrix}\vspace{0.2 cm}
{\bf y}^{(n)}_1   \\ \vspace{0.2 cm}
{\bf y}^{(n)}_2   \\ \vspace{0.2 cm}
\vdots   \\ \vspace{0.2 cm}
{\bf y}^{(n)}_{{\rm{m_r}}}   \\
\end{pmatrix},~
{\bf W}^{(n)} \triangleq \begin{pmatrix}\vspace{0.2 cm}
{\bf w}^{(n)}_1   \\ \vspace{0.2 cm}
{\bf w}^{(n)}_2   \\ \vspace{0.2 cm}
\vdots   \\ \vspace{0.2 cm}
{\bf w}^{(n)}_{{\rm{m_r}}}   \\
\end{pmatrix}
\end{equation}
\noindent where ${\bf y}^{(n)}_j$, ${\bf w}^{(n)}_j$, and ${\bf s}^{(n)}_i$ denote the received vector at the $j$-th receive antenna, the noise vector at the $j$-th receive antenna, and the transmitted vector  from the $i$-th transmit antenna, respectively, and are given as   

   \begin{align} \label{eq: sys mod y w s anttena i}
    {\bf y}^{(n)}_j &\triangleq [ \hspace{0.1 cm}
y^{(n)}_j[0]   \hspace{0.2 cm}
y^{(n)}_j[1]    \hspace{0.2 cm}
\cdots    \hspace{0.2 cm}  y^{(n)}_j [{n_{\rm{s}}-1}] \hspace{0.1 cm}]^{\rm{T} }\\
    {\bf w}^{(n)}_j &\triangleq [ \hspace{0.1 cm} w^{(n)}_j[0]     \hspace{0.2 cm} 
w^{(n)}_j[1] \hspace{0.2 cm}
\cdots    \hspace{0.2 cm}
w^{(n)}_j [{n_{\rm{s}}-1}] \hspace{0.1 cm} ]^{\rm{T}}\\
    {\bf s}^{(n)}_i &\triangleq [ \hspace{0.1 cm} s^{(n)}_i[0]     \hspace{0.2 cm} 
s^{(n)}_i[1] \hspace{0.2 cm}
\cdots    \hspace{0.2 cm}
s^{(n)}_i [{n_{\rm{s}}-1}] \hspace{0.1 cm} ]^{\rm{T}} \nonumber \\
    &\hspace{-0.4 cm} =[ \hspace{0.1 cm}
 x^{(n)}_i(0)  \hspace{0.2 cm}
x^{(n)}_i(T_{\rm{sa}})   \hspace{0.2 cm}
\cdots   \hspace{0.2 cm} ... 
x^{(n)}_i((n_{\rm{x}}-1) T_{\rm{sa}})   \hspace{0.2 cm}  \underbrace{ 0 \hspace{0.2 cm} \cdots   \hspace{0.2 cm} 0 \hspace{0.1 cm} }_{n_{\rm{z}}} ]^{\rm{T}} 
\end{align}

Note that if the receiver starts to receive samples before any data is transmitted, it only receives noise samples, thus, we have ${\bf Y}^{(n)}= {\bf W}^{(n)},~ n<0$ in \eqref{Sys Model: matrix form conv 2}. In order to avoid \ac{isi}, the length of the zero-padding should be greater than or equal to the number of channel taps, i.e. $n_{\rm{z}} \ge n_{\rm{h}}$. This assumption holds throughout our analysis in this paper.  We also consider the Class A impulsive noise model, i.e. Gaussian mixtures, that is defined as  \cite{shongwe2015study}
\begin{equation} \label{eq: noise pdf}
f_{W^{(n)}_j[k]}(w) = \sum^{L-1}_{l=0} p_l \mathcal{CN}(w:0, \sigma_{{\rm{w}}_l}^2),
\end{equation}
\noindent for $k \in \{0,1,\cdots,{n_{\rm{s}}-1} \}$. It is well-known that the Gaussian mixture noise is a more accurate noise model than the conventional simple Gaussian model \cite{wang1999robust}. That is, the Gausssian mixture distribution models the noise in many real-world channels such as urban, underwater, and indoor channels,  more accurately compared to the conventional simple Gaussian model \cite{blackard1993measurements, middleton1973man, middleton1993elements,middleton1987channel}.

According to the \ac{clt}, transmitted \ac{ofdm} samples, i.e. $s^{(n)}_i[k]= x^{(n)}_i(kT_{{\rm{sa}}}),~ \forall k \in \{0, 1, \cdots , {n_{\rm{x}}}-1 \}$,  can be modeled as \ac{iid} zero-mean Gaussian random variables. Hence,
\begin{align} \label{eq: ofdm samples gauss}
s^{(n)}_i[k] ~\text{or}~ x^{(n)}_i(kT_{{\rm{sa}}}) \sim \mathcal{CN}(0,\sigma^2_x),~~~ \forall k \in \{0, 1, \cdots , {n_{\rm{x}}}-1 \}
\end{align}
where 
\begin{align} \label{eq: ofdm samples gauss power}
\mathbb{E}\Big{\{}s^{(n)}_i[k] s^{(n)}_p[k']^*\Big{\}}&=\mathbb{E}\Big{\{}x^{(n)}_i(kT_{{\rm{sa}}}) x^{(n)}_p(k'T_{{\rm{sa}}})^*\Big{\}} \\ 
&= \sigma^2_x \delta[k-k'] \delta[i-p],\\ \nonumber &\forall i,p \in \{1,2,\cdots,{m_{\rm{t}}} \}, \\  \nonumber
&\forall k,k' \in \{0, 1, \cdots , {n_{\rm{x}}}-1 \}.
\end{align}
Now, assume that  there is a \ac{to} $\tau \triangleq dT_{\rm{sa}}+\epsilon$ between the transmitter and the receiver, where  $d$ and $\epsilon$ represent the integer and fractional part of the \ac{to}, respectively. Since the fractional part of \ac{to}, $\epsilon$,  can  be  corrected  through channel  equalization, it  suffices to estimate the  beginning  of  the  received  \ac{ofdm} vector within one sampling period. 
In fact, it  is  common in practice to model the \ac{to} as a multiple of  the  sampling   period,  and consider the remaining fractional error as part of the channel impulse response. To this end, we focus on estimating the integer part of the \ac{to},  $d$, which is essential in order to perform the \ac{fft} operation at the receiver, and decode the data in subsequent steps.

The next sections propose an approximate \ac{ml} estimator for estimating $d$.

\section{Maximum Likelihood Estimation For Single-Input Single-Output System} \label{sec: siso}

In order to better understand the main ideas, and for the sake of notation simplicity, we first derive the approximate \ac{ml} \ac{to} estimator for a \ac{siso}-\ac{ofdm} system where ${\rm{m_t}}=1$ and ${\rm{m_r}}=1$. We then extend the proposed results for \ac{siso} systems in order to obtain the \ac{ml} \ac{to} estimator for \ac{mimo}-\ac{ofdm} systems.

%\subsection{Single-Input Single-Output System} \label{subsec: siso}

We consider a \ac{siso} wireless system; hence, ${\rm{m_t}}=1$ and ${\rm{m_r}}=1$. For the sake of notation simplicity, we remove the subscripts $i$  or $j$, denoting the variables associated with the $i$-th transmit or $j$-th receive antenna, from ${\bf y}^{(n)}_j, y^{(n)}_j[k], {\bf w}^{(n)}_j, w^{(n)}_j[k], {\bf s}^{(n)}_i, s^{(n)}_i[k], {\bf {\rm{H}}}_{ji}, h_{ji}[k]$, and $ x^{(n)}_i(k)
$. Hence, Equation \eqref{Sys Model: matrix form conv 2} can be rewritten as 
\begin{align}\label{siso: matrix form conv 2}
{\bf y}^{(n)}=
\begin{cases}
 {\rm{H}} {\bf s}^{(n)} + {\bf w}^{(n)} \triangleq  {\bf v}^{(n)} +{\bf w}^{(n)},  \,\,\,\,\,\,\,\,\,\,\,\ n \ge 0  \\
{\bf w}^{(n)}, \,\,\,\,\,\,\,\,\,\,\,\,\,\,\,\,\,\,\,\,\,\,\,\,\,\,\,\,\,\,\,\,\,\,\,\,\,\,\,\,\,\,\,\,\,\,\,\,\,\,\,\,\,\,\,\,\,\,\,\,\,\,\,\,\,\,\,\,\,\,\,\,\,\,  n<0
\end{cases}
\end{align}
\noindent where
   \begin{align}
    {\bf y}^{(n)} &\triangleq [ \hspace{0.1 cm}
y^{(n)}[0]    \hspace{0.2 cm}
y^{(n)}[1]    \hspace{0.2 cm}
\cdots    \hspace{0.2 cm}  y^{(n)} [{n_{\rm{s}}}-1] \hspace{0.1 cm}]^{\rm{T}} \\    {\bf v}^{(n)} &\triangleq [ \hspace{0.1 cm} v^{(n)}[0]     \hspace{0.2 cm} \cdots 
v^{(n)}[1]    \hspace{0.2 cm}
\cdots    \hspace{0.2 cm}
v^{(n)} [{n_{\rm{s}}}-1] \hspace{0.1 cm} ]^{\rm{T}}\\
    {\bf w}^{(n)} &\triangleq [ \hspace{0.1 cm} w^{(n)}[0]     \hspace{0.2 cm} \cdots 
w^{(n)}[1]    \hspace{0.2 cm}
\cdots    \hspace{0.2 cm}
w^{(n)} [{n_{\rm{s}}}-1] \hspace{0.1 cm} ]^{\rm{T}}\\
    {\bf s}^{(n)} &\triangleq [ \hspace{0.1 cm} s^{(n)}[0]     \hspace{0.2 cm} 
s^{(n)}[1] \hspace{0.2 cm}
\cdots    \hspace{0.2 cm}
s^{(n)} [{n_{\rm{s}}-1}] \hspace{0.1 cm} ]^{\rm{T}} \nonumber \\
    &\hspace{-0.4 cm} =[ \hspace{0.1 cm}
 x^{(n)}(0)  \hspace{0.2 cm}
x^{(n)}(T_{\rm{sa}})   \hspace{0.2 cm}
\cdots   \hspace{0.2 cm} 
x^{(n)}((n_{\rm{x}}-1) T_{\rm{sa}})   \hspace{0.2 cm}  \underbrace{ 0 \hspace{0.2 cm} \cdots   \hspace{0.2 cm} 0 \hspace{0.1 cm} }_{n_{\rm{z}}} ]^{\rm{T}}.
\end{align}
\noindent and  ${\bf {\rm{H}}}$ is the lower triangular $n_{\rm{s}} \times n_{\rm{s}}$ Toeplitz channel matrix with first column  $[h[0] \ h[1] \ \cdots \  h[n_{\rm{h}}-1] ~ \  0 \  \cdots \ 0]^\text{T}$.

Now, we assume that the integer part of the \ac{to}, $d$, can take  values from a set $\mathcal{D}= \{-n_{\rm{s}}+1,\cdots,-1,0,1,\cdots,n_{\rm{s}}-1\}$, i.e. $d \in \mathcal{D}$. Note that the negative values of the delay, $d$, denotes the time when the receiver starts early to receive samples. That is, for $d<0$, the receiver receives  $\lvert d\rvert$ noise samples from the environment, and then receives the transmitted \ac{ofdm}  samples starting from the $\lvert d \rvert+1$-th sample at the receiver. Similarly, when $d \ge 0$, the receiver starts late to receive the samples. In other words,  the receiver misses the first  $d$ samples from the transmitted \ac{ofdm} samples. Allowing $d$ to take both negative and positive values enables the final estimator to perform both the frame and symbol synchronization; hence, it possesses a significant advantage. 

The problem of estimating the \ac{to} can be formulated as a multiple hypothesis testing problem, i.e. ${\rm{H}}_d,~ \forall d \in \mathcal{D}=\{-n_{\rm{s}}+1, \cdots,  n_{\rm{s}}-1\}$. We first assume that the receiver uses $N$ observation vectors, ${\bf y}^{(0)}$, ${\bf y}^{(1)}$, $\cdots$, ${\bf y}^{(N-1)}$, each with length $n_{{\rm{s}}}$, in order to estimate the timing offset $d$. Later, we allow the receiver to use any arbitrary number of received samples, not necessarily a multiple of $n_{{\rm{s}}}$, for estimation. In order to derive the \ac{ml} \ac{to} estimator, we need to obtain the joint \ac{pdf} of the $N$ observation vectors under the different hypotheses ${\rm{H}}_d$. Assuming ${\rm{H}}_{d}$, this \ac{pdf} is denoted by $f( {\bf y}^{(0)}, {\bf y}^{(1)}, \cdots, {\bf y}^{(N-1)}| {\rm{H}}_{d})$.

The following lemma from \cite{koosha2020} gives some insights about the samples of the received vectors. We provide this lemma here without the proof.
\begin{lemma} \label{lemma: obser vec}
 For any arbitrary $n$ and $k,~ 0 \le k \le n_{\rm{s}}-1$, let us define the actual index of an observation sample, i.e. $y^{(n)}[k]$, as $k^{\rm{ac}}= n_{\rm{s}}n+k$. Then, for any arbitrary $d \in \mathcal{D}$, the elements of the  observation vectors, i.e. $y^{(n)}[k]$, are uncorrelated random variables. Moreover, the elements with an actual index difference greater than $n_{\rm{h}}$ are independent. 
\end{lemma}

\begin{corollary} \label{corol: independ}
Samples from different observation vectors are independent. 
\end{corollary}

According to Lemma \ref{lemma: obser vec}, although the elements of the observation vectors, i.e $y^{(n)}[k]$, are uncorrelated, and those with actual index difference greater than $n_{\rm{h}}$ are independent, but, they are not generally independent. Hence, deriving a closed-form expression for the joint \ac{pdf}, i.e. $f( {\bf y}^{(0)}, {\bf y}^{(1)}, \cdots, {\bf y}^{(N-1)}| {\rm{H}}_{d})$, is not mathematically tractable. However, authors in \cite{koosha2020} showed that the final estimator is less sensitive to the independency assumption of the elements of the observation vectors  than their sole \ac{pdf}s. Thus, using Corollary \ref{corol: independ}, we can write 
\begin{equation} \label{eq: indepency 1}
\begin{split} 
    f( {\bf y}^{(0)}, {\bf y}^{(1)}, \cdots, &{\bf y}^{(N-1)}| {\rm{H}}_{d})  \\ 
    &= f_{{\bf Y}^{(0)}}( {\bf y}^{(0)}|{\rm{H}}_{d}) ~ f_{{\bf Y}^{(1)}}({\bf y}^{(1)}|{\bf y}^{(0)},{\rm{H}}_{d})
\cdots   \\ 
  &~~~~~~~~~ f_{{\bf Y}^{(N-1)}}({\bf y}^{(N-1)} | {\bf y}^{(0)}, {\bf y}^{(1)}, \cdots , {\rm{H}}_{d})\\ 
  &= \prod_{n=0}^{N-1}  f_{{\bf Y}^{(n)}}( {\bf y}^{(n)}| {\rm{H}}_{d}) \\
    &\simeq \prod_{n=0}^{N-1}
\prod_{k=0}^{n_{\rm{s}}-1}  f_{Y^{(n)}[k]}(y^{(n)}[k]| {\rm{H}}_{d})
\end{split}
\end{equation}
Let us denote the in-phase and quadrature components of a received sample, i.e. $y^{(n)}[k]$, by $y^{(n)}_{{\rm{I}}}[k]$ and $y^{(n)}_{{\rm{Q}}}[k]$, respectively. Since, the in-phase and quadrature components are independent from each other, one can rewrite Equation \eqref{eq: indepency 1} as 
\begin{equation} \label{eq: indepency 2}
\begin{split} 
    f( {\bf y}^{(0)}, &{\bf y}^{(1)}, \cdots, {\bf y}^{(N-1)}| {\rm{H}}_{d}) \simeq \prod_{n=0}^{N-1}
\prod_{k=0}^{n_{\rm{s}}-1}  f_{Y^{(n)}[k]}(y^{(n)}[k]| {\rm{H}}_{d})\\
&= \prod_{n=0}^{N-1}
\prod_{k=0}^{n_{{\rm{s}}}-1}  f_{Y^{(n)}_{{\rm{I}}}[k]}\big(y^{(n)}_{{{\rm{I}}}}[k] |  {\rm{H}}_{d} \big) ~  f_{Y^{(n)}_{{\rm{Q}}}[k]}\big(y^{(n)}_{{{\rm{Q}}}}[k] |  {\rm{H}}_{d} \big)
\end{split}
\end{equation}
Using Equation \eqref{siso: matrix form conv 2},  we have $Y^{(n)}_{{{\rm{I}}}}[k]=V^{(n)}_{{{\rm{I}}}}[k] + W^{(n)}_{{{\rm{I}}}}[k]$. The same equation holds for the quadrature components of the received samples. Hence, 
the \ac{pdf} $f_{Y^{(n)}_{{\rm{I}}}[k]}\big(y^{(n)}_{{{\rm{I}}}}[k] |  {\rm{H}}_{d} \big)$, or similarly  $f_{Y^{(n)}_{{\rm{Q}}}[k]}\big(y^{(n)}_{{{\rm{Q}}}}[k] |  {\rm{H}}_{d} \big)$, for $n\ge0$ can be obtained as 
\begin{equation} \label{eq: convol siso}
\begin{split}
f_{Y^{(n)}_{{\rm{I}}}[k]}\big(y^{(n)}_{{{\rm{I}}}}[k] |  {\rm{H}}_{d} \big)=& f_{V^{(n)}_{{{\rm{I}}}}[k]} \big( v^{(n)}_{{{\rm{I}}}}[k] |  {\rm{H}}_{d} \big) \\ 
&* f_{W^{(n)}_{{{\rm{I}}}}[k]}\big( w^{(n)}_{{{\rm{I}}}}[k] |  {\rm{H}}_{d} \big)
\end{split}
\end{equation}
\noindent where $*$ denotes the convolution operation. Deriving the \ac{pdf} $f_{V^{(n)}_{{{\rm{I}}}}[k]} \big( v^{(n)}_{{{\rm{I}}}}[k] |  {\rm{H}}_{d} \big)$, or $f_{V^{(n)}_{{{\rm{Q}}}}[k]} \big( v^{(n)}_{{{\rm{Q}}}}[k] |  {\rm{H}}_{d} \big)$, is complex, and results in a complicated expression \cite{koosha2020}. Here, we approximate this \ac{pdf} with a Gaussian \ac{pdf}  with zero mean and  a specific variance for each received sample in favor of reducing the complexity of the final estimator. We later justify this assumption in more details.

%Note that the Gaussian assumption is effective in this scenario, and in fact is a "first-cut" analysis for practical algorithm design. 

The first step to derive the \ac{pdf} of in-phase component of a received sample $Y^{(n)}_{{\rm{I}}}[k]$ is to determine the corresponding mean and variance of $V^{(n)}_{{\rm{I}}}[k]$ under assumption ${\rm{H}}_d$.  In the absence of \ac{to}, i.e. ${\rm{H}}_0$,  Equation \eqref{siso: matrix form conv 2} holds and can be expanded as Equation \eqref{eq: expa conv 2} given at the top of the next  page.  

Using Equations \eqref{eq: expa conv 2} and \eqref{eq: ofdm samples gauss} and the fact that $\mathbb{E}\{h_{{\rm{I}}}[k]\}=\mathbb{E}\{h_{{\rm{Q}}}[k]\}=0$, we conclude 
\begin{align}\label{eq: mean siso}
\mathbb{E}\{V^{(n)}_{{\rm{I}}}[k]|{\rm{H}}_0\}=0.
\end{align}
By substituting
\begin{align}
\mathbb{E}\big{\{}s^{(n)}_{{{\rm{I}}}^2}[k]\big{\}}&= \mathbb{E}\big{\{}s^{(n)}_{{{\rm{Q}}}^2}[k]\big{\}} \\ \nonumber
&=\begin{cases}
\frac{\sigma_{\rm{x}}^2}{2},  \,\,\,\,\,\,\,\,\,\,\,\,\,\,\,\,\,\,\,\,\,\,\  0 \le k \le n_{\rm{x}}-1  \\
0, \,\,\,\,\,\,\,\,\,\,\,\,\,\,\,\,\,\,\,\,\,\,\,\,\,\,\,\,  n_{\rm{x}} \le k \le n_{\rm{s}}-1
\end{cases}
\end{align}
\noindent and 
$\mathbb{E}\{h^{(n)}_{\rm{I}}[k] h^{(n)}_{\rm{I}}[r]\}=\mathbb{E}\{h^{(n)}_{\rm{Q}}[k] h^{(n)}_{\rm{Q}}[r]\}=\sigma_{{\rm{h}}_k}^2\delta [k-r]$, and using Equation \eqref{eq: expa conv 2}, one can derive 
\begin{align}\label{eq: var siso}
\sigma^2_k &\triangleq   \mathbb{E}\big{\{} {\big( v^{(n)}_{{\rm{I}}}}[k] \big)^2|{\rm{H}}_0\big{\}}  \\
     &=\begin{cases}
      \sum^{b}_{r=a} \frac{ \sigma^2_{h_r}   \sigma^2_{{\rm{x}}}}{2} &~~~ 0 \le k < n_{\rm{x}} + n_{\rm{h}}-2 \\ \\ \nonumber
       0    &~~~       n_{\rm{x}} + n_{\rm{h}}-1 \le k \le n_{\rm{s}}-1  \\ 
       &~~~ \rm{or}~ n<0 \\
     \end{cases}
\end{align}
where
 \begin{equation} \label{eq: expa conv 21}
(a,b)  \hspace{-0.2em}= \hspace{-0.2em}
     \begin{cases}
      (0,k) &~ 0 \le k \le n_{\rm{h}}-2\\
       (0,n_{\rm{h}}-1) &~ n_{\rm{h}}-1 \le k \le n_{\rm{x}}-1\\
       (k-n_{\rm{x}}+1,n_{\rm{h}}-1) &~ n_{\rm{x}} \le k \le n_{\rm{x}}+n_{\rm{h}}-2.\\
     \end{cases}
\end{equation}

Using Equations \eqref{eq: var siso} and \eqref{eq: expa conv 21}, one can define the vector of variances of the received samples under assumption ${\rm{H}}_0$ as

%%%%%%%%%%%%%%%%%%%%%%%%%%%%%%%%%%%%%%%%%%%%%%%%%%%%%%%%%%%%%%%%%%% top page equation
\begin{figure*}[t] %% over both columns
\label{eq:6767}

 \begin{equation} \label{eq: expa conv 2}
v^{(n)}_{{\rm{I}}}[k] \hspace{-0.2em}= \hspace{-0.2em}
     \begin{cases}
      \sum^{m}_{u=0} h_{\rm{I}}[u] s^{(n)}_{{\rm{I}}}[k-u] - h_{\rm{Q}}[u] s^{(n)}_{{\rm{Q}}}[k-u]  ~~~~~~~~~~~~~~~~~~~~~ 0 \le k < n_{\rm{h}}-2 \\
       \sum^{n_{\rm{h}}-1}_{u=0} h_{\rm{I}}[u] s^{(n)}_{{\rm{I}}}[k-u] - h_{\rm{Q}}[u] s^{(n)}_{{\rm{Q}}}[k-u]  ~~~~~~~~~~~~~~~~~~~~ n_{\rm{h}}-1 \le k \le n_{\rm{x}}-1, \\
       \sum^{n_{\rm{h}}-1}_{u=m-n_{\rm{x}}+1} h_{\rm{I}}[u] s^{(n)}_{{\rm{I}}}[k-u] - h_{\rm{Q}}[u] s^{(n)}_{{\rm{Q}}}[k-u]  ~~~~~~~~~~~~~ n_{\rm{x}} \le k \le n_{\rm{x}}+n_{\rm{h}}-2,\\
       0 ~~~~~~~~~~~~~~~~~~~~~~~~~~~~~~~~~~~~~~~~~~~~~~~~~~~~~~~~~~~~~~~~~~~~~~ n_{\rm{x}}+n_{\rm{h}}-1 \le k \le n_{\rm{s}}-1,
     \end{cases}
\end{equation}
\\
\noindent\rule{\textwidth}{1pt}
\end{figure*}
%%%%%%%%%%%%%%%%%%%%%%%%%%%%%%%%%%%%%%%%%%%%%%%%%%%%%%%%%%%%%%%%%%% top page equation

%%%%%%%%%%%%%%%%%%%%%%%%%%%%%%%%%%
%%%%%%%%%%%%%%%%%%%%%%%%%%%%%%%%%%
$\boldsymbol{ \sigma}^2_{{\bf V}_{\rm{I}}|{\rm{H}}_0}\triangleq$
\begin{equation}\label{eq: var matrix H0 siso} 
\hspace{-0.5cm}   \left[\begin{array}{@{}c}
 \vspace{0.2cm}  
 \vdots\\ \vspace{0.2cm}
\sigma^2_{{\bf V}_{\rm{I}}[-2]|{\rm{H}}_0} \\ \vspace{0.2cm} \sigma^2_{{\bf V}_{\rm{I}}[-1]|{\rm{H}}_0}  \\   \hdashline  \\
 \vspace{0.2cm} \sigma^2_{{\bf V}_{\rm{I}}[0]|{\rm{H}}_0} \\  \sigma^2_{{\bf V}_{\rm{I}}[1]|{\rm{H}}_0}  \\ \vdots \\ \vspace{0.2cm} \sigma^2_{{\bf V}_{\rm{I}}[n_{\rm{s}}-1]|{\rm{H}}_0}  \\   \hdashline  \\
 \vspace{0.2cm} \sigma^2_{{\bf V}_{\rm{I}}[n_{\rm{s}}]|{\rm{H}}_0} \\  \sigma^2_{{\bf V}_{\rm{I}}[n_{\rm{s}}+1]|{\rm{H}}_0}  \\ \vdots \\ \vspace{0.2cm} \sigma^2_{{\bf V}_{\rm{I}}[c_1+2n_{\rm{s}}-1]|{\rm{H}}_0}    \\
 \hdashline \vspace{-0.3cm} \\
\vdots \vspace{0.3cm}
 \end{array}\right] \stackrel{(a)}{=} %split vectors here
\hspace{-0.1cm} 
   \left[\begin{array}{@{}c@{}}
 \vspace{0.2cm}  
 \vdots\\
0 \\ \vspace{0.2cm} 0  \\     \hdashline  \\
 \vspace{0.2cm} \sigma^2_{0} \\  \sigma^2_{1}  \\ \vdots \\ \vspace{0.2cm} \sigma^2_{n_{\rm{s}}-1}   \\ \hdashline  \\
 \vspace{0.2cm} \sigma^2_{0} \\  \sigma^2_{1}  \\ \vdots \\ \vspace{0.2cm} \sigma^2_{n_{\rm{s}}-1}   \\
 \hdashline \vspace{-0.3cm} \\
\vdots \vspace{0.3cm}
\end{array}\right]
\end{equation}

%%%%%%%%%%%%%%%%%%%%%%%%%%%%%%%%%
%%%%%%%%%%%%%%%%%%%%%%%%%%%%%%%%%
\noindent where $\boldsymbol{ \sigma}^2_{{\bf V}_{\rm{I}}|{\rm{H}}_0}[k]= \sigma^2_{{\bf V}_{\rm{I}}[k]|{\rm{H}}_0}$ corresponds to the variance of the  $k$-th received sample under ${\rm{H}}_0$, i.e.  $Y^{(n)}_{{\rm{I}}}[k]$.

Now, note that any \ac{to} in the received samples of length $m$ results in an only a shift in the vector of variances given in \eqref{eq: var matrix H0 siso} as 
\begin{equation} \label{eq: var Hd siso}
\boldsymbol{ \sigma}^2_{{\bf V}_{\rm{I}}|{\rm{H}}_d} =    \boldsymbol{ \sigma}^2_{{\bf V}_{\rm{I}}|{\rm{H}}_0} (d:d+m-1)
\end{equation}
where we denote a shortened version of an arbitrary unlimited-length vector ${\bf a}= [\cdots~a[-2]~a[-1]~a[0]~a[1]~a[2]~ \cdots ]^\textrm{T}$ by ${\bf a}(r:k),~ r\le k$, and is defined as
\begin{equation} \label{eq: def trun i j}
    {\bf a}(r:k) \triangleq \big[a[r]~a[r+1]~ \cdots~ a[k]\big]^\textrm{T}.
\end{equation}
Finally, using Equations \eqref{eq: var Hd siso} and \eqref{eq: mean siso}, we have 
\begin{align}\label{eq: V siso}
V^{(n)}_{{\rm{I}}}[k]\sim \mathcal{N}(0, \boldsymbol{ \sigma}^2_{{\bf V}_{\rm{I}}|{\rm{H}}_d}[k]),
\end{align}
\noindent where $V^{(n)}_{{\rm{I}}}[k]\sim \mathcal{N}(0, 0) $ means $V^{(n)}_{{\rm{I}}}[k]=0$ . Similar equation holds for $V^{(n)}_{{\rm{Q}}}[k]$.

Now that we derived the variance of  $V^{(n)}_{{\rm{I}}}[k]$ under assumption ${\rm{H}}_d$, i.e. Eq. \eqref{eq: var Hd siso}, we are ready to derive the \ac{pdf} of an in-phase component of a received sample $Y^{(n)}_{{\rm{I}}}[k]$ under ${\rm{H}}_d$. Since  $W^{(n)}[k]=W^{(n)}_{{\rm{I}}}[k]+j W^{(n)}_{{\rm{Q}}}[k]$, from \eqref{eq: noise pdf} we have 
\begin{equation}\label{eq: noise in-phase siso}
f_{W^{(n)}_{\rm{I}}[i]}(w)= \sum^{L-1}_{l=0} p_l \mathcal{CN}\Big{(}w:0, \frac{\sigma_{{\rm{w}}_l}^2}{2}\Big{)}.
\end{equation}
Finally, by substituting \eqref{eq: V siso} and \eqref{eq: noise in-phase siso} into \eqref{eq: convol siso}, we can write
\begin{equation} \label{eq: pdf inphase siso}
\begin{split}
&f_{Y_{\rm{I}}[k]} (y|{\rm{H}}_d) =
\sum^{L-1}_{l=0}
\frac{p_l}{\sqrt{2\pi \frac{\sigma_{{\rm{w}}_l}^2}{2}}}
\frac{1}{\sqrt{2\pi \boldsymbol{ \sigma}^2_{{\bf V}_{\rm{I}}|{\rm{H}}_d}[k]}} \times \\
&\int^{\infty}_{-\infty}
\exp\bigg\{-\frac{1}{2\frac{\sigma_{{\rm{w}}_l}^2}{2}}   \big(y -  v\big)^2 \bigg\} \exp\Big{(}-\frac{v^2}{2\boldsymbol{ \sigma}^2_{{\bf V}_{\rm{I}}|{\rm{H}}_d}[k]}\Big{)} dv\\
&=
\sum^{L-1}_{l=0}
\frac{p_l \exp\big{(}-\frac{y^2}{\sigma_{{\rm{w}}_l}^2}\big{)}}{\sqrt{2\pi \frac{\sigma_{{\rm{w}}_l}^2}{2}}}
\frac{1}{\sqrt{2\pi \boldsymbol{ \sigma}^2_{{\bf V}_{\rm{I}}|{\rm{H}}_d}[k]}} \\
&\int^{\infty}_{-\infty}
\exp\bigg\{- \Big{(} \frac{1}{\sigma_{{\rm{w}}_l}^2} +  \frac{1}{2\boldsymbol{ \sigma}^2_{{\bf V}_{\rm{I}}|{\rm{H}}_d}[k]} \Big{)} v^2 + \Big{(}\frac{2y}{\sigma_{{\rm{w}}_l}^2}\Big{)} v \bigg\}  dv\\
&\stackrel{a}{=} \sum^{L-1}_{l=0}    \frac{p_l}{\sqrt{2\pi (\boldsymbol{ \sigma}^2_{{\bf V}_{\rm{I}}|{\rm{H}}_d}[k] + \frac{\sigma_{{\rm{w}}_l}^2}{2})}}  \exp\bigg{(}-\frac{y^2}{2\Big{(}\boldsymbol{ \sigma}^2_{{\bf V}_{\rm{I}}|{\rm{H}}_d}[k] + \frac{\sigma_{{\rm{w}}_l}^2}{2}\Big{)}}\bigg{)}
\end{split}
\end{equation}
\noindent one can derive the similar equations for $f_{Y_{\rm{Q}}[k]} (y|{\rm{H}}_0)$.

Before we dig into the last step and derive the \ac{pdf} of $Y^{(n)}[k]$, we use the next definitions to simplify notations. 

\textit{Definition 1}: The mathematical function $B\big(z,t \big): \mathds{R} \times \mathds{R} \longmapsto \mathds{R}$  is defined as 
\begin{equation}
B\big(z,t \big)= 
    \sum^{L-1}_{l=0}    \frac{p_l}{\sqrt{2\pi (t + \frac{\sigma^2_{{\rm{w}}_l}}{2})}}  \exp\Big{(}-\frac{z^2}{2(t + \frac{\sigma^2_{{\rm{w}}_l}}{2})}\Big{)},
\end{equation}
\noindent where  $t \ge 0$.

\textit{Definition 2}: The mathematical function $\mathcal{P}\big(c, \kappa \big) : \mathds{C} \times \mathds{R} \longmapsto \mathds{R}$ is called $\mathcal{P}$-function, and  is defined as 
\begin{equation} \label{eq: P-function def}
\mathcal{P}\big(c, \kappa \big)= B\big(c_1,\kappa \big) B\big(c_2,\kappa \big),
\end{equation}
\noindent where $c_1= {\rm{Re}}\{c\}$, $c_2=
{\rm{Im}}\{c\}$, and $\kappa \ge  0$ is called shape parameter.

Using \textit{Definitions 1} and \textit{2} and Equation \eqref{eq: pdf inphase siso}, one can show that the \ac{pdf}  of the received samples under ${\rm{H}}_d$, i.e. $f_{Y[k]} (y|{\rm{H}}_d)$,  are $\mathcal{P}$-functions with different shape parameters as 
\begin{equation} \label{eq: repeat pat n>0}
\begin{split}
 f_{Y[k]} (y|{\rm{H}}_d) &= f_{Y_{\rm{I}}[k]} (y_{\rm{I}}|{\rm{H}}_d)  f_{Y_{\rm{Q}}[k]} (y_{\rm{Q}}|{\rm{H}}_d)\\
 &= B\big(y_{\rm{I}}, \boldsymbol{ \sigma}^2_{{\bf V}_{\rm{I}}|{\rm{H}}_d}[k] \big) B\big(y_{\rm{Q}}, \boldsymbol{ \sigma}^2_{{\bf V}_{\rm{I}}|{\rm{H}}_d}[k] \big) \\
 &= \mathcal{P}\big(y,\boldsymbol{ \sigma}^2_{{\bf V}_{\rm{I}}|{\rm{H}}_d}[k] \big),
 \end{split}
\end{equation}
\noindent where $y_{\rm{I}}= {\rm{Re}}\{y\}$, $y_{\rm{Q}}=
{\rm{Im}}\{y\}$. 

Finally, for any arbitrary number of received samples $m \ge 1$, i.e. ${\bf y}= \big[~ y[0]~y[1]~\cdots~y[m-1]~\big]^{\rm{T}}$, using \eqref{eq: repeat pat n>0}, we have 
\begin{equation} \label{eq: ml siso}
\begin{split}
\hat{\boldsymbol{d}}^{\text{opt}} &= \operatorname*{argmax}_{d \in \mathcal{D} }~  \prod^{m-1}_{k=0} f_{Y[k]}({\bf y}[k]| {\rm{H}}_{d})\\
&= \operatorname*{argmax}_{d \in \mathcal{D} }~  \prod^{m-1}_{k=0}  \mathcal{P}\big({\bf y}[k], \boldsymbol{ \sigma}^2_{{\bf V}_{\rm{I}}|{\rm{H}}_d}[k] \big).
\end{split}
\end{equation}

The next theorem summarizes our discussions in this section.

\begin{theorem}
In a doubly selective channel and under impulsive noise defined in \eqref{eq: noise pdf}, the approximate \ac{nda} \ac{ml} \ac{to} estimator for a \ac{zp}-\ac{ofdm} systems in given as

\begin{equation} \label{eq: theo ml siso}
\hat{\boldsymbol{d}}^{\text{opt}} = \operatorname*{argmax}_{d \in \mathcal{D} }~  \prod^{m-1}_{k=0}  \mathcal{P}\big({\bf y}[k], \boldsymbol{ \sigma}^2_{{\bf V}_{\rm{I}}|{\rm{H}}_d}[k] \big).
\end{equation}
where $\mathcal{P}\big(c, \kappa \big)$ is defined in  \eqref{eq: P-function def} and $\boldsymbol{ \sigma}^2_{{\bf V}_{\rm{I}}|{\rm{H}}_d}[k]$ is given in \eqref{eq: var Hd siso}.

\end{theorem}

\subsection{Higher Order Statistical  Analysis}

In order to see that the Gaussian approximation of 
$V^{(n)}_{{{\rm{I}}}}[k]$, or $V^{(n)}_{{{\rm{Q}}}}[k]$, is mainly valid, we have plotted the true \ac{pdf} of $Y^{(n)}_{{{\rm{I}}}}[k]$ versus the approximate \ac{pdf} given in \eqref{eq: pdf inphase siso} in Fig. \ref{fig: Empirical vs Analytical} for different values of $k=1$, chosen from the first range given in \eqref{eq: expa conv 2}, and $k=150$, chosen from the second range in \eqref{eq: expa conv 2}. Moreover, the corresponding probabilistic measures such as mean, variance, skewness, and kurtosis of the true and approximate \ac{pdf}s are given in Table \ref{table: pdfs metrics}. Skewness measures the asymmetry of a probability distribution around its mean, and is defined as 
\begin{align}
sk \triangleq \frac{\mathbb{E}\{(Y-\mu)^3\}}
{\mathbb{E}\{(Y-\mu)^2\}}.
\end{align}
Kurtosis is a measure of sharpness of a \ac{pdf}, and is defined as 
\begin{align}
ku \triangleq \frac{\mathbb{E}\{(Y-\mu)^4\}}
{\big(\mathbb{E}\{(Y-\mu)^2\} \big)^2}.
\end{align}

The true \ac{pdf} is obtained through $10^6$ Monte Carlo simulations where $n_{\rm{x}}=512, n_{\rm{z}}=20, n_{\rm{h}}=10, f_{\rm{s}}=10^6$, and modulation order is $M=128$. Also, an exponential power delay profile is considered and is given as 
\begin{align}
\mathbb{E}\{|h[k]|^2\}= \sigma^2_{{\rm{h}}_k}=\alpha \exp\big{(}-\beta k\big{)},
\, k=0,1,\cdots, 9,
\end{align}
\noindent where $\alpha=0.396$ is a normalizing factor that guarantees the average energy of the channel taps sum to one, and $\beta=0.5$. For the sake of simplicity in illustrations, we set $p_0=1$, and $p_l=0,~ \forall l\neq 0$, which resembles a Gaussian noise.

As can be seen in Fig. \ref{fig: Empirical vs Analytical}(a), when $k=1$, the approximate \ac{pdf}  deviates from the true \ac{pdf} in terms of Kurtosis, the sharpness of the \ac{pdf}, and almost matches the other metrics such as skewness and variance. This deviation holds for the samples in the first and third ranges given in \eqref{eq: expa conv 2} for the given setup, i.e.  $0 \le k \le n_{\rm{h}}-2$ and $n_{\rm{x}} \le k \le n_{\rm{x}}+n_{\rm{h}}-2$. However, as seen in Fig. \ref{fig: Empirical vs Analytical}(b), when $k=150$, the approximate \ac{pdf}  reasonably matches with the true \ac{pdf} in terms of all probabilistic measures given in Table \ref{table: pdfs metrics}. This approximate matching holds for the second and fourth ranges given in \eqref{eq: expa conv 2} for the given setup, i.e. $n_{\rm{h}}-1 \le k \le n_{\rm{x}}-1$ and 
$n_{\rm{x}}+n_{\rm{h}}-1 \le k \le n_{\rm{s}}-1$. These ranges contain $\frac{n_{\rm{x}} +n_{\rm{z}}-2n_{\rm{h}}+2   }{n_{\rm{x}}+n_{\rm{z}} }$ of the samples in a given received \ac{ofdm} vector which is more than $96\%$ of the samples. That is, the approximate \ac{pdf}  reasonably matches the true \ac{pdf} of more than $95\%$ of the samples in a given received \ac{ofdm} vector while it slightly differs for less than $5\%$ of the received samples. This implies that the performance of the final estimator should not degrade considerably; although, the final estimator should have significantly lower complexity compared to that of given in \cite{koosha2020} due to the simplicity of the approximate \ac{pdf}.

\begin{figure}
\centering
\subfloat[]{\includegraphics[width=3.1in]{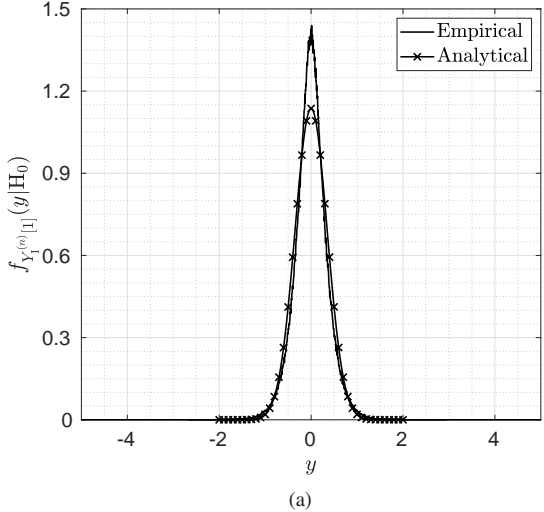}} 
\newline
\subfloat[]{\includegraphics[width=3.1in]{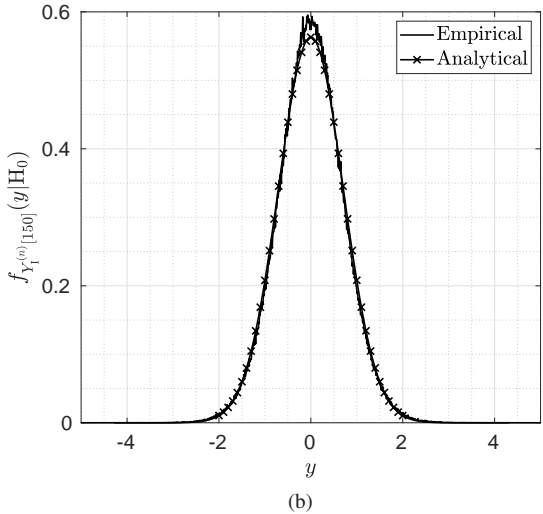}}
\caption{The comparison between the empirical and analytical \ac{pdf}s of the received samples for $k=1$ and $k=150$ when \ac{snr}$=15$ dB.} 
\label{fig: Empirical vs Analytical}
\end{figure}

%\begin{figure}[t!] 
%\centering
%\begin{subfigure}
% \subfloat{{\includegraphics[height=2.835in]{pdf_indx_1_snr_15.eps} }}%
% \subcaption{(a)}
%\end{subfigure}
%\begin{subfigure}
%\includegraphics[height=2.835in]{pdf_indx_150_snr_15.eps}
%\subcaption{(b)}
%  \end{subfigure}
%  \caption{The comparison between the empirical and analytical \ac{pdf}s of the received samples for $k=1$ and $k=150$ when \ac{snr}$=15$ dB. }\label{fig: Empirical vs Analytical}
%\end{figure}

\begin{table}[t!]
\centering 
 \caption{Various probability measures for  empirical \ac{pdf} versus analytical \ac{pdf}. }
\label{table: pdfs metrics}
\resizebox{0.5\textwidth}{!}{ \begin{tabularx}{0.42\textwidth}{@{  }l*{5}{C}@{  }}
\toprule
& \multicolumn{2}{c}{$k = 1$ } & \phantom{abc}& \multicolumn{2}{c}{$k = 150$}  \\\cmidrule{2-3} \cmidrule{5-6} & Empirical & Analytical  && Empirical & Analytical \\ 
\midrule
Mean           & -0.00022 & 0 &  & 0.00041 & 0 \\ 
Variance       & 0.1231 & 0.1232 &  & 0.5021  & 0.5023  \\
Skewness       & -0.0048 & 0 &  & 0.0090 & 0  \\ 
Kurtosis       & 4.4519 & 3 &  & 3.3000 & 3   \\ 
\bottomrule
\end{tabularx} }
\end{table}

% top page equation
%%%%%%%%%%%%%%%%%%%%%%%%%%%%%%%%%%%%%%%%%%%%%%%%%%%%%%%%%%%%%%%%%%

\subsection{Case Study: Gaussian noise when $d \ge 0$ }

Assume the receiver starts late to receive the samples. That is $d \ge 0$. Also, assume the receiver uses $N$ vectors of $n_{\rm{s}}$ samples in order to estimate the \ac{to} in an environment with Gaussian noise. Taking logarithm of \eqref{eq: ml siso}, and after some mathematical manipulations, we have   
\begin{equation} \label{eq: ml siso wed}
\begin{split}
\hat{\boldsymbol{d}}^{\text{opt}} &= 
 \operatorname*{argmin}_{d \in \mathcal{D} }~  \sum^{Nn_{\rm{s}}-1}_{k=0}   \frac{|{\bf y}[k]|^2}{\boldsymbol{ \sigma}^2_{{\bf V}_{\rm{I}}|{\rm{H}}_d}[k] + \frac{\sigma_{{\rm{w}}_l}^2}{2}}.
\end{split}
\end{equation}
Equation \eqref{eq: ml siso wed} shows that the final estimator in the given setup turns into a \ac{wed}. Note that \ac{wed} assigns larger weights to those samples carrying no information, i.e. noise samples, while dedicating smaller weights to those carrying information.

\subsection{Energy Detector}

Inspired by \ac{wed} derived in \eqref{eq: ml siso wed}, we introduce a sub-optimal lower complexity estimator, i.e. \ac{ed}, where we consider an extreme case of \ac{wed}. That is,  we assign a very large weight, i.e. one, to noise samples while we assign a zero weight to those samples carrying information. Hence, 
\begin{equation} \label{eq: ml siso ed}
\begin{split}
\hat{\boldsymbol{d}}^{\text{opt}} &= 
 \operatorname*{argmin}_{d \in \mathcal{D} }~ \sum^{N-1}_{r=0} \sum^{\phi_2(d)}_{k=\phi_1(d) }   |{\bf y}[k]|^2, 
\end{split}
\end{equation}
where 
\begin{align}
&\phi_1(d) \triangleq n_{\rm{x}}+ n_{\rm{h}}-1 + r n_{\rm{s}}-d    \\
&\phi_2(d) \triangleq (r+1)n_{\rm{s}}-1-d.
\end{align}

In the next section, we extend the \ac{ml} \ac{to} estimator given in Equation \eqref{eq: ml siso} for \ac{siso} systems to an \ac{ml} \ac{to} estimator for \ac{mimo} systems.

\subsection{Maximum Likelihood Estimation for Multiple-Input Multiple-Output System} \label{sec: mimo}

In this subsection, we derive the \ac{ml} \ac{to} for a \ac{mimo}-\ac{ofdm} system under a frequency selective channel. We use the results proposed for \ac{siso} scenario, and try to extend the notations.  The next theorem gives  \ac{ml} \ac{to} estimator for \ac{mimo} systems. 

\begin{theorem} \label{theo: pdf y mimo}
For a \ac{zp} \ac{mimo}-\ac{ofdm} wireless communication system  and Class A impulsive noise defined in \eqref{eq: noise pdf}, the approximate \ac{ml} \ac{to} estimator is given as
\begin{equation} \label{eq: ml mimo}
\hat{\boldsymbol{d}}^{\text{opt}} 
= \operatorname*{argmax}_{d \in \mathcal{D} }~ \prod^{m_{\rm{r}}}_{j=1} 
\prod^{m-1}_{k=0}  \mathcal{P}\big({\bf y}_j[k],  \boldsymbol{\sigma}^2_{{\bf V}_{\rm{I}}|{\rm{H}}_d}[k] \big).
\end{equation}
where $\mathcal{P}\big(c, \kappa \big)$ is defined in  \eqref{eq: P-function def} and $\boldsymbol{ \sigma}^2_{{\bf V}_{\rm{I}}|{\rm{H}}_d}[k]$ is given in \eqref{eq: var Hd siso}. 

\end{theorem}

\begin{proof}
Using Equations \eqref{Sys Model: matrix form conv 2}, \eqref{matrix H} and \eqref{sys mod:  s y w} for a \ac{mimo}-\ac{ofdm} system, we have 
\begin{equation} \label{eq: proof mimo}
{\bf y}_{j}^{(n)}= \sum_{i=1}^{{\rm{m_t}}} {\bf {\rm{H}}}^{(n)}_{ji} {\bf s}^{(n)}_i + {\bf w}^{(n)}_j,~ \forall j \in \{1,2,\cdots,{\rm{m_r}}\}.
\end{equation}
\noindent Also, we have 
\begin{align} \label{eq: ofdm samples gauss power}
\mathbb{E}\Big{\{}s^{(n)}_i[k] s^{(n)}_p[k']^*\Big{\}}&=\mathbb{E}\Big{\{}x^{(n)}_i(kT_{{\rm{sa}}}) x^{(n)}_p(k'T_{{\rm{sa}}})^*\Big{\}} \\  \nonumber
&= \frac{\sigma^2_x}{m_{\rm{t}}} \delta[k-k'] \delta[i-p],\\ \nonumber &\forall i,p \in \{1,2,\cdots,{\rm{m_t}} \}, \\  \nonumber
&\forall k,k' \in \{0, 1, \cdots , {n_{\rm{x}}}-1 \}.
\end{align}
\noindent since the transmit power should remain the same for \ac{siso} and \ac{mimo} systems. 
Using \eqref{eq: proof mimo} and following the same steps as in \ac{siso} case, one can  arrive at \eqref{eq: ml mimo}.
\end{proof}

Next section compares the complexity of the proposed approximate \ac{ml} algorithm with that of given in \cite{koosha2020}.

\section{Complexity} 
\label{sec: complexity}

Complexity of an algorithm plays a crucial role in using  the algorithm in wireless communication systems. That is, an algorithm should be rather simple yet accurate enough in order to be considered for practical implementations. To this end, the complexity of the proposed approximate \ac{ml} algorithm, referred to as A-\ac{ml}, with the complexity of the original \ac{ml} estimator in \cite{koosha2020}, denoted as  O-\ac{ml}, and that of Transition Metric \cite{LeNir2010} are given in Table  \ref{table:complexity}.    As can be seen, A-\ac{ml} has significantly lower computational complexity compared to O-\ac{ml} while possessing a negligible performance loss in terms of lock-in probability, shown in next section. Also, A-\ac{ml} has  the same computational complexity as Transition Metric while demonstrating a large performance gap in terms of lock-in probability. 

%Table \ref{table:complexity} shows the total number of real floating-point operations including addition, subtraction, multiplication, and division.

As seen, the main advantage of A-\ac{ml}  is its simplicity which enables the designer to easily extend it to \ac{mimo} systems.  Note that A-\ac{ml}  can be implemented in fully vectorized format that makes it considerably faster than O-\ac{ml}. One can further improve the complexity and the exhaustive search in \eqref{eq: ml mimo} by using complexity-reduced search algorithms such as Golden Section Search algorithm which has a significantly lower complexity of $\mathcal{O}(\log(|\mathcal{D}|))$ compared to that of exhaustive search,  $\mathcal{O}(|\mathcal{D}|)$.

\begin{table}[t!]                    
 \centering 
    \caption{Complexity of the proposed algorithms}
    \resizebox{0.39\textwidth}{!}{ \begin{tabularx}{0.3\textwidth}{cc}  
    \toprule
    %\multicolumn{3}{c}{\tableHeader{Complexity}} \\
    %\midrule
    
    {\bf Estimator} & {\bf Computational Complexity }   \\
    \hline
    
    A-ML   
    &    $\mathcal{O}(Nn_{\rm{s}})$ \\ 
      O-ML & $\mathcal{O}(Nn_{\rm{s}}^3)$      \\ 
            Transition Metric & $\mathcal{O}(Nn_{\rm{s}})$     \\ 
     
    \bottomrule
    \end{tabularx} }

    \label{table:complexity}                           
\end{table}

\section{Simulations}

\label{sec: simul}

In this section, we compare the proposed algorithm  A-\ac{ml} with O-\ac{ml} given in \cite{koosha2020}, and Transition Metric \cite{LeNir2010}.

\subsection{Simulation Setup}

Unless otherwise mentioned, the following setup is considered for simulations. A \ac{zp}-\ac{ofdm} system with 128-QAM modulation in a frequency-selective Rayleigh fading channel is considered with data samples length of $n_{\rm{x}}=512$, and \ac{zp} guard interval of length $n_{\rm{z}}=20$.  The number of received \ac{ofdm} vectors used for estimation is set $N=10$. Sampling rate is $f_{\rm{s}}=10^6$. The exponential channel delay profile parameters are $\alpha=1, \beta=0.05$, where $n_{\rm{h}}=10$.  A Jakes model for Doppler spectrum with maximum Doppler shift of  $f_{\rm{D}}=5$ Hz is considered. A two-components impulsive noise with parameters  $p_0=0.99$, $p_1=0.01$, $\sigma_{{\rm{w}}_0}^2=1$, and $\sigma_{{\rm{w}}_1}^2=100$ is set.  \ac{snr}  in dB is defined as 
$\gamma\triangleq 10 \log(\frac{\sigma_{ \rm{x}}^2}{\sigma_{\rm{w}}^2})$. Simulation results are obtained through $10^4$ Monte Carlo realizations, and the delay is uniformly chosen from the range   $d \in [-30 ,   30]$.

%\textcolor{red}{The total transmit power is normalized to one}.

\subsection{Simulation Results}

The probability of lock-in of A-\ac{ml}, O-\ac{ml} and Transition Metric for different values of \ac{snr} when $m_{\rm{t}}=m_{\rm{r}}=1$ are depicted in Fig. \ref{fig: snr}. As shown, there is a negligible performance gap between A-\ac{ml} and O-\ac{ml} while A-\ac{ml} possesses a much lower computational complexity. Also, A-\ac{ml} significantly outperforms Transition Metric.

\begin{figure}
\centering
\includegraphics[height=2.835in]{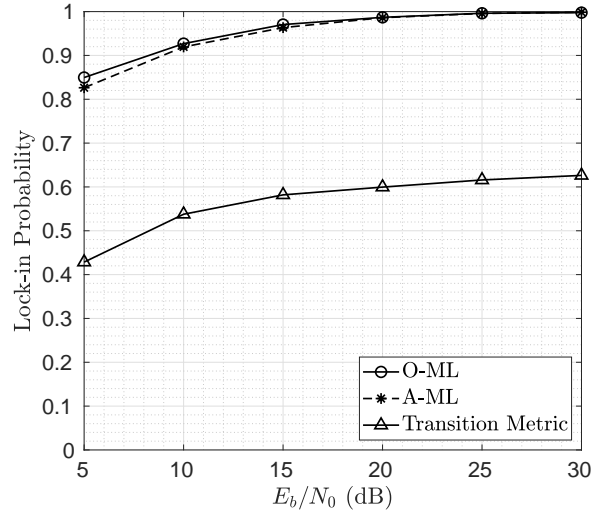}
  \caption{Lock-in probability of A-ML, O-ML and Transition Metric  for  different  values  of  \ac{snr}.  } \label{fig: snr}
\end{figure}

The performance of A-\ac{ml} versus the number of observation vectors used for estimation is shown in Fig. \ref{fig: obser}. As seen, the performance of A-\ac{ml} improves as the number of observation vectors increases. This figure shows that with a reasonable buffer capacity or an increase in the number of antennas,  a receiver using A-\ac{ml} is able to achieve high lock-in probability, e.g more than 0.9.

\begin{figure}
\centering
\includegraphics[height=2.835in]{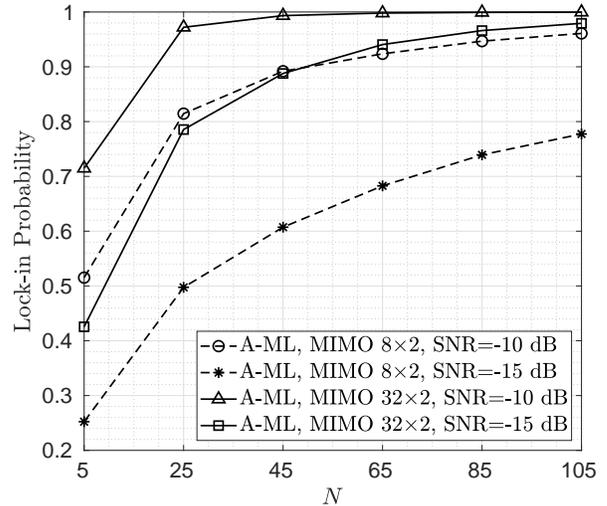}
  \caption{Lock-in probability of  A-\ac{ml} for different number of observation vectors used for estimation.  } \label{fig: obser}
\end{figure}

The performance of \ac{wed} and  \ac{ed} for different values of positive delays, i.e. when the receiver starts late to receive samples,  and  for various \ac{snr} values is depicted in Fig. \ref{fig: wed}. This figure shows that \ac{ed} is able to achieve a high probability of lock-in in higher \ac{snr}s while a considerable performance gap is clear at lower \ac{snr}s. This performance gap is originating from the simplifying assumption of zero and one weight assignments in \ac{ed}, which results in significantly lower computational complexity of \ac{ed} compared to \ac{wed}.

\begin{figure}
\centering
\includegraphics[height=2.835in]{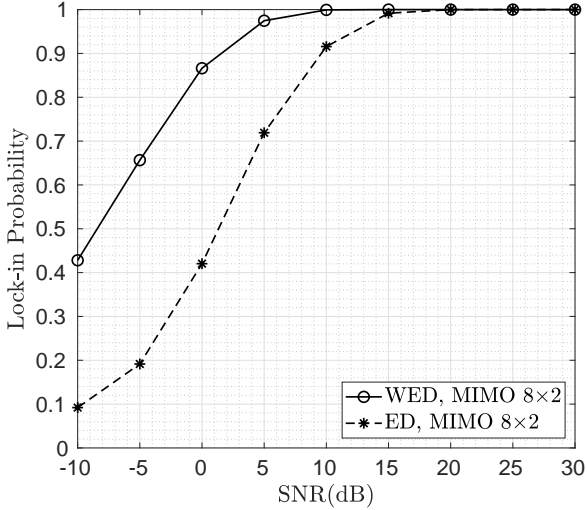}
  \caption{Lock-in probability of  \ac{wed} and \ac{ed} for values of \ac{snr}.  } \label{fig: wed}
\end{figure}

Fig. \ref{fig: mimo} shows the performance of A-\ac{ml} for \ac{mimo} systems for different values of $m_{\rm{t}}$ and $m_{\rm{r}}$. As seen, the performance of A-\ac{ml} improves significantly due to the fact that the number of receive samples used for estimation increases when the number of transmit or receive antenna increases. Increasing the number of observations improves the accuracy of \ac{ml} estimators; hence, the performance of A-\ac{ml} improves.

\begin{figure}
\centering
\includegraphics[height=2.835in]{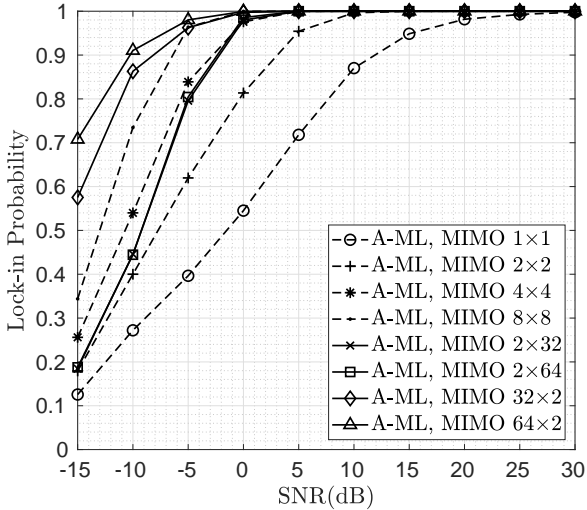}
  \caption{Lock-in probability of  A-\ac{ml} for various number of antennas.  } \label{fig: mimo}
\end{figure}

The effect of impulsive noise on the performance of A-\ac{ml} is shown in Fig. \ref{fig: noise}. Here, we consider a two-components Gaussian mixture where $\sigma_{{\rm{w}}_0}^2=1$ $\sigma_{{\rm{w}}_1}^2=100$, and we change the ratio of $p_0/p_1$. As seen, when a component, i.e. $\sigma_{{\rm{w}}_0}^2=1$, becomes strong, i.e. $p_0$ increases, the performance improves. This is because the uncertainty of A-\ac{ml} arising from Eq. \eqref{eq: pdf inphase siso} decreases as $p_0$ increases.

\begin{figure}
\centering
\includegraphics[height=2.835in]{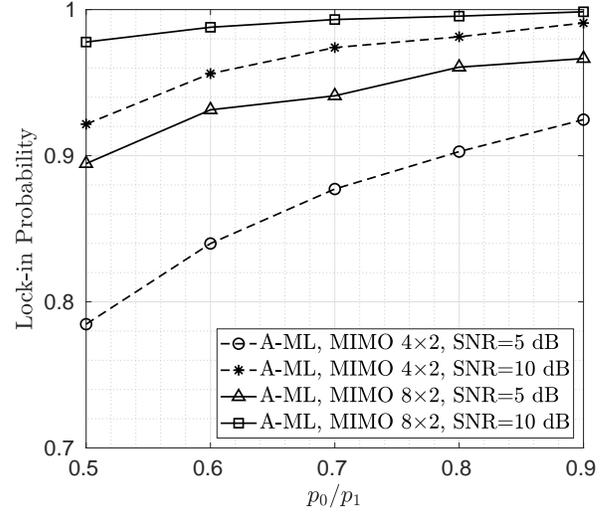}
  \caption{Lock-in probability of  A-\ac{ml} for different values of $p_0/p_1$.  } \label{fig: noise}
\end{figure}

Since A-\ac{ml} employs  power delay profile for synchronization,  we study the sensitivity of A-\ac{ml} to power delay profile estimation errors in Fig. \ref{fig: sens}. In order to generate power delay profile errors, we use the following equation 
\begin{equation}
\sigma^2_{{\rm{h}}^{\rm{new}}_k} =   (1+ A_k \alpha ) \sigma^2_{{\rm{h}}_k}, \, k=0,1,\cdots, n_{\rm{h}}-1,
\end{equation}
where $A_k$ is uniformly (randomly) chosen for each tap  from the set $\{ -1, 1\}$. $\sigma^2_{{\rm{h}}^{\rm{new}}_k} $, instead of $\sigma^2_{{\rm{h}}_k}$, is then fed to A-ML in order to estimate \ac{to}. The probability of lock-in of A-\ac{ml} for different values of $\alpha$ is shown in Fig. \ref{fig: sens}. Although the performance of A-\ac{ml} degrades with an increase in power delay profile error; however, a large error is required to achieve a $50\%$ loss in terms of lock-in probability. Moreover, note that an error of 0.7 results in less than four  percent performance loss which implies that A-\ac{ml} is fairy insensitive to power delay profile estimation errors.

\begin{figure}
\centering
\includegraphics[height=2.835in]{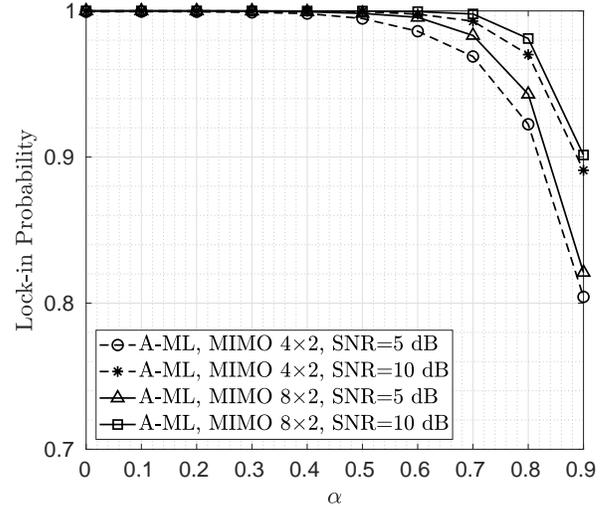}
  \caption{Lock-in probability of  A-\ac{ml} for different values of power delay profile estimation errors.  } \label{fig: sens}
\end{figure}

\section{Conclusion}
\label{sec: conclu}

\ac{zp}-\ac{ofdm} systems possess many advantages compared to \ac{cp}-\ac{ofdm} systems. However, the time synchronization in \ac{zp}-\ac{ofdm} systems are significantly challenging  due to the lack of \ac{cp}. In this paper, we proposed an approximate yet accurate low-complexity \ac{nda} \ac{ml} \ac{to} estimator, i.e. A-\ac{ml}, for \ac{zp} \ac{mimo}-\ac{ofdm} systems in highly selective channels. We showed that A-\ac{ml} has a significantly lower complexity than that of proposed in \cite{koosha2020}, i.e. O-\ac{ml}, while having a negligible performance gap in terms of lock-in probability. This makes A-\ac{ml}, unlike O-\ac{ml}, suitable for practical implementations.  Moreover, it is shown that A-\ac{ml}  dramatically outperforms the current stat-of-the-art \ac{nda} \ac{to} estimator for \ac{zp}-\ac{ofdm} referred to as Transition Metric.

%\vspace{0.5cm}
%**********************************************
%\\

%\vspace{0.5cm}
%**********************************************
%\\
%\vspace{5cm}

\IEEEpeerreviewmaketitle

\bibliographystyle{IEEEtran}
\bibliography{IEEEabrv,Reference}

\end{document}